\documentclass[twocolumn,showpacs,preprintnumbers,amsmath,amssymb]{revtex4-1}

\usepackage{amsfonts}
\usepackage{amsmath}
\usepackage{amssymb}
\usepackage{amsthm}
\usepackage{graphicx}
\usepackage{dcolumn}
\usepackage{bm}
\usepackage[all]{xy}

\begin{document}
\title{\Large{\textbf{Algebraic characterization of binary graphs}}}

\author{Tiziano Squartini\\
\emph{Instituut-Lorentz for Theoretical Physics, Leiden Institute of Physics,}\\
\emph{University of Leiden, Niels Bohrweg 2, 2333 CA Leiden (The Netherlands)}\\
\emph{E-mail: squartini@lorentz.leidenuniv.nl}
}

\date{\today}

\begin{abstract}
One of the fundamental concepts in the statistical mechanics field is that of ``ensemble''. Ensembles of graphs are collections of graphs, defined according to certain rules. The two most used ensembles in network theory are the microcanonical and the grandcanonical (whose definitions mimick the classical ones, originally proposed by Boltzmann and Gibbs), even if the latter is far more used than the former to carry on the analytical calculations. For binary (undirected or directed) networks, the grandcanonical ensemble is defined by considering all the graphs with the same number of vertices and a variable number of links, ranging from $0$ to the maximum: $N(N-1)/2$ for binary, undirected graphs and $N(N-1)$ for binary, directed graphs. Even if it is commonly used almost exclusively as a tool to calculate the average of some topological quantity of interest, its structure is so rich to deserve an analysis on its own. In this paper a logic-algebraic characterization of the grandcanonical ensemble of binary graphs is provided.
\end{abstract}

\maketitle

\section{Introduction}

The so-called complex networks theory is now an accepted branch of physics and, in pariticular, of statistical physics. In fact, networks are mainly studied by means of those tools developed in the traditional field of statistical mechanics and this has lead to the definition of concepts and methods formally analogue to those used there \cite{erdos,gilbert,bender,bollobas,barabaus,new,calda,cl,new2,likel}. One of the most successful concepts is the one of \emph{ensemble}, whose definition mimicks the traditional one: a set of networks each one having an associated probability coefficient, used to compute the average of some topological property of interest \cite{huang}. 

The less statistically-biased way to build an esemble is by means of the MaxEnt procedure, that is the maximization of the Shannon entropy under some imposed contraints \cite{shannon,jaynes}: a number of different ensembles can be defined by varying the number and the type of constraints. The main difficulty with this procedure is the calculation of the partition function for such ensembles, leading to a closed-form function: however, if linear constraints are chosen, a smart choice of the support greatly simplifies the calculations \cite{likel}. The far most common choice is to consider the set of all graphs with the same number of nodes (say $N$) and with a number of links varying from zero to the maximum. In what follows we will consider only binary networks, so that only two possibilities are given for each pair of nodes: they are connected (and there is one link between them) or not (and zero links are present between them). This implies that the maximum number of links is $N(N-1)/2$ for binary, undirected networks (the total number of undirected pairs, given $N$ objects) and $N(N-1)$ for binary, directed networks (the total number of directed pairs, given $N$ objects).

If we consider the total number of nodes' pairs as the available volume to store the given particles (that is, the directed or undirected links) we can identify the collection of graphs defined above as the usual \emph{grandcanonical ensemble} in physics, characterized by the three constraints $(V,\:\mu,\:T)$, where $V$ is the volume (the total number of nodes' pairs), $\mu$ is the chemical potential (controlling for the average number of particles/links) and $T$ is the temperature (controlling for the average value of the energy: a physical definiton of energy adaptable to networks - even if maybe unnecessary - is still missing but attempts to include it in the set of quantites imposable as networks contraints have already been made \cite{temp}). Clearly, as long as the volume of the system does not change and the number of particles changes, the only meaningful identification is with the grandcanonical ensemble (in fact, neither the microcanonical, or $(V,\:N,\:E)$, nor the canonical, or $(V,\:N,\:T)$, ensemble satisfy these requests \cite{jaynes,shannon}).

It turns out that allowing for the number of links to vary, greatly simplifies the calculations involving the partition function, thus making the grandcanonical ensemble a powerful tool to solve analytical models. However, despite the easiness in carrying on the calculations on this support, no efforts have been made to characterize it from a mathematical point of view. 

In the next sections, an aswer to this question will be provided (and with an abuse of language, we will call ``ensemble'' only the support of the distribution, disregarding the type of distribution defined over it).

\section{Explicit generation of the grandcanonical ensemble}

The explicit generation of the grandcanonical ensemble of binary networks can be realized by a tree-like process. Let us consider binary, undirected networks for simplicity. The pairs of nodes can be double-indexed (i.e. $ij$, with $i<j$ and ignoring, as usual in network theory, the diagonal elements) and listed according to a double-order: firstly, by ordering the index $i$ and, then, the index $j$. For example, with $N=4$ nodes we would have the following $N(N-1)/2=6$ pairs:

$$
ij=12,\:13,\:14,\:23,\:24,\:34.
$$

These undirected pairs will be the vertices of our tree. In $N(N-1)/2=6$ steps we will build all the branches of our tree, whose ``leaves'' will be the matrices constituting the grandcanonical ensemble. This procedure, even if long, makes the mechanism leading to the definition of this kind of ensemble clear. Note that the double index leads to the definition of a matrix (the adjacency matrix), usually indicated with $A$ and whose elements, the $a_{ij}$s, represent the pairs of nodes under consideration: $a_{ij}=1$ indicates that the vertices $i$ and $j$ are linked; $a_{ij}=0$ indicates that the vertices $i$ and $j$ are not linked.

Fig. \ref{tree} shows the first four (out of six) steps leading to the definition of the grandcanonical ensemble of binary networks with $N=4$. Let us start from the first pair of nodes and decide if the vertices $1$ and $2$ are linked or not ($a_{12}=1$ or $a_{12}=0$), the two choices corresponding to the first two branches. After the first pair, let us consider the next one, $ij=13$, and ask again if the nodes $1$ and $3$ are linked or not; these two possibilities correspond to the two pairs of branches drawn under the vertex $ij=13$ (and note the there are two vertices of this kind, corresponding to the previous, two different choices). So, the recipe is clear: it prescribes to draw two branches under each vertex (corresponding to 0 and 1) to link it to the following, until the last choice.

Proceeding in this way, we have explicitly listed all the possibilities, by assigning $0$ or $1$ to each adjacency matrix element and producing $2^{\frac{N(N-1)}{2}}=2^{6}=64$ binary, undirected matrices. Note that the empty graph, $E_{4}$, is produced by a series of $N=6$ zeros (on the left side of the tree) and that the complete graph, $K_{4}$, is produced by a series of $N=6$ ones (on the right side of the tree).

\begin{figure}[h!]
\includegraphics[scale=0.8]{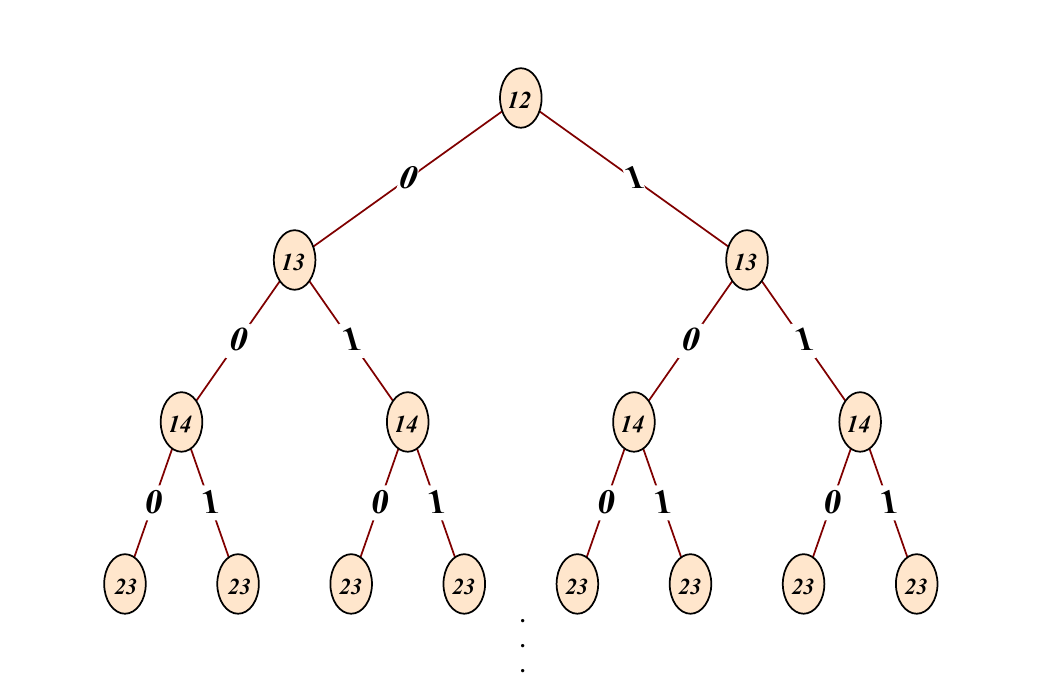}
\caption{The grandcanonical tree for binary, undirected networks.}
\label{tree}
\end{figure}

A worked example of this tree-generation is provived in fig. \ref{grand}. The same procedure can be applied to the binary, directed networks. In this case, the tree would have $2^{N(N-1)}$ leaves because also the ``reciprocal'' pairs should be considered and drawn, accordingly (so, each pair $ij$, with $i<j$, would be linked to the reciprocal, $ij$ with $i>j$, and not directly to $i+1\:j$, with $i+1<j$).

\begin{figure}[h!]
\includegraphics[scale=0.8]{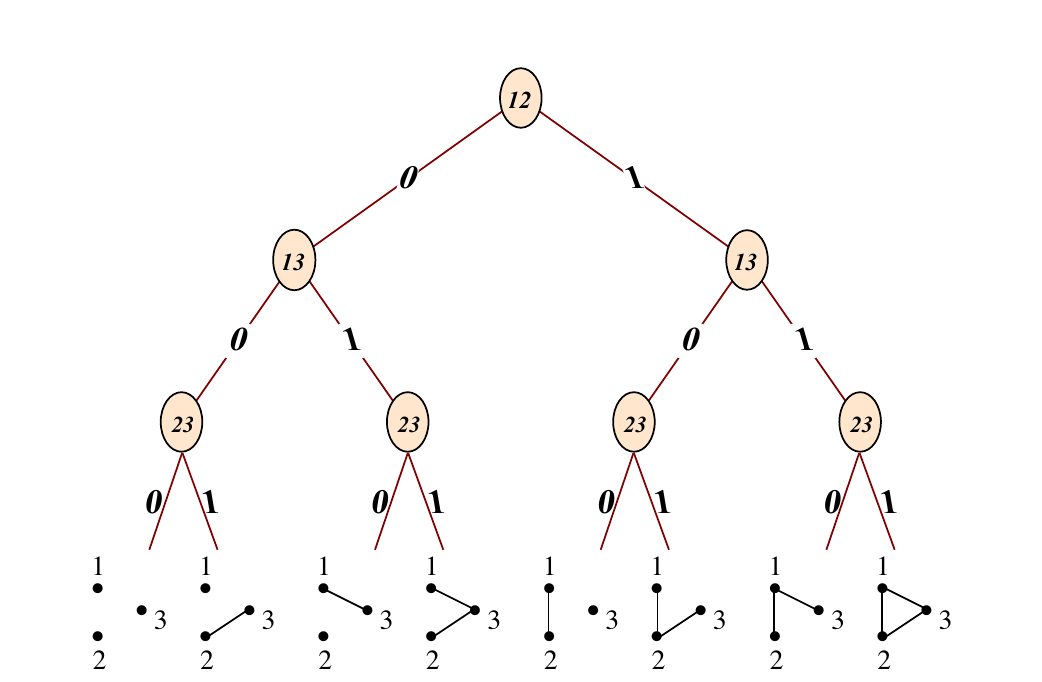}
\caption{The grandcanonical tree of binary, undirected networks with $N=3$ vertices. The leaves are the binary, undirected 
networks constituting the corresponding grandcanonical ensemble, whose cardinality is $2^{3}=8.$}
\label{grand}
\end{figure}

\section{The grandcanonical ensemble as a set of relations}

A binary graph is usually defined as a couple of sets, $V$ and $E$ where $V$ is the set of vertices and $E$ is the set of edges: when considering binary, undirected networks $E\subseteq[V]^2$ (the latter being the subset of the power set of $V$, whose elements have cardinality two); when considering binary, directed networks $E\subseteq V\times V$ (the latter being the cartesian product of $V$ with itself). These definitions can be unified, considering the definition of \emph{relation} \cite{DCM}. 

\newtheorem{rel}{Definition}
\begin{rel}
Given two sets, $S$ and $T$, a relation $\mathcal{R}$ is a subset of the cartesian product $S\times T$.
\end{rel}

Let us come back to the explicit generation of the grandcanonical ensemble. We could also ask, at every step: \emph{are $i$ and $j$ in relation?} By repeating this question at every step of the previous tree, and by answering ``yes'' (1 - and drawing a link between the corresponding vertices) or ``not'' (0 - and not drawing a link between the corresponding vertices), a graphical representation of the considered relation is obtained.

In fact, let us consider the cartesian product $V\times V$ and think about the link set, $E$: links are nothing but the pairs $(v_i,\:v_j)$ where $i, j=1\dots N$ and $v_{i}\in V,\:\forall\:i$. So, $a_{ij}=0$ indicates that the vertices $i$ and $j$ are not linked and the pair $(v_{i},\:v_{j})$ does not belong to $E$: $(v_{i},\:v_{j})\not\in E$; on the other side, $a_{ij}=1$ indicates that the vertices $i$ and $j$ are linked and the pair $(v_{i},\:v_{j})$ belongs to $E$: $(v_{i},\:v_{j})\in E$. In summary,

$$
(v_i\mathcal{R}v_j)\Leftrightarrow a_{ij}=1\:\mbox{and}\:(v_i\overline{\mathcal{R}}v_j)\Leftrightarrow a_{ij}=0,\:\forall\:i, j
$$

\noindent (where the line over the letter $\mathcal{R}$ indicates that the two elements are not related) showing that \emph{binary graphs are graphical representations of relations} \cite{DCM}. So, given $N$ vertices, the total number of binary relations (binary networks) definable over them is $2^{N^2}$. Even if the correspondence between graphs and relations is well-known, complex network theory can provide the tools to ``quantify'' the considered relation.

\subsection{Reflexivity}

\newtheorem{rel2}{Definition}
\begin{rel}
A relation is said to be reflexive when $(v_{i}\mathcal{R}v_{i}),\:\forall\:i$.
\end{rel}

To verify if a given relation is reflexive \cite{DCM}, we have to verify if every vertex is in relation with itself. This is easily done by considering the trace of the adjacency matrix $A$: $\mbox{Tr}(A)=\sum_{i=1}^N a_{ii}$. 

We have three cases: all the vertices are in relation with themselves, only some of them are in relation with themselves and no vertex is in relation with itself; the trace assumes three different values, respectively:

\begin{displaymath}
\left\{ \begin{array}{ll}
\mbox{Tr}(A)=0: & (v_{i}\overline{\mathcal{R}}v_{i}),\:\forall\:i;\:\:\:\mbox{\emph{antireflexive}}.\\
0<\mbox{Tr}(A)<N: & \exists\:i:\:(v_{i}\overline{\mathcal{R}}v_{i});\:\mbox{\emph{not-reflexive}}.\\
\mbox{Tr}(A)=N: & (v_{i}\mathcal{R}v_{i}),\:\forall\:i;\:\:\:\mbox{\emph{reflexive}}.
\end{array} \right.
\end{displaymath}

It is possibile to count how many relations (networks) are of the three kinds, above: there are $2^{N(N-1)}$ antireflexive relations, $2^{N(N-1)}$ reflexive relations and the remaining $2^{N^2}\left[1-\frac{1}{2^{N-1}}\right]$ are not-reflexive. If, as usual in network theory, the diagonal elements (or self-loops) are often ignored: this means that $a_{ii}=0,\:\forall\:i$ and only the $2^{N(N-1)}$ antireflexive relations are considered (exactle the cardinality of the grandcanonical ensemble of binary, directed networks).

\subsection{Symmetry}

\newtheorem{rel3}{Definition}
\begin{rel}
A relation is said to be symmetric when $(v_{i}\mathcal{R}v_{j})\Rightarrow (v_{j}\mathcal{R}v_{i}),\:\forall\:i,j$. A relation is said to be antisymmetric when $(v_{i}\mathcal{R}v_{j})\Rightarrow (v_{j}\overline{\mathcal{R}}v_{i}),\:\forall\:i,j$ with $i\neq j$.
\end{rel}

In other words, to verify the symmetry of a relation we have to verify if every link has its reciprocal (i.e. another link connecting the same two vertices but pointing in the opposite direction) \cite{DCM}. As for the reflexivity, a quantity exists to measure the degree of (a)symmetry of a given relation: the reciprocity \cite{new3},

$$
r=\frac{\sum_{i(\neq j)}\sum_{j}a_{ij}a_{ji}}{\sum_{i(\neq j)}\sum_{j}a_{ij}},
$$

\noindent (where every index runs from 1 to $N$) according to which

\begin{displaymath}
\left\{ \begin{array}{ll}
r=0: & (v_{i}\mathcal{R}v_{j})\Rightarrow (v_{j}\overline{\mathcal{R}}v_{i}),\:\forall\:i\neq j;\hfill\mbox{\emph{antisymm}}.\\
0<r<1: & \exists\:\:i,j,\:i\neq j:\:(v_{i}\mathcal{R}v_{j})\Rightarrow (v_{j}\overline{\mathcal{R}}v_{i});\:\mbox{\emph{asymm}}.\\
r=1: & (v_{i}\mathcal{R}v_{j})\Rightarrow (v_{j}\mathcal{R}v_{i}),\:\forall\:i, j;\hfill\mbox{\emph{symmetric}}.
\end{array} \right.
\end{displaymath}

Note that an empty network has all the adjacency matrix entries equal to zero. This implies that $r=0/0$, that is $r$ is an indeterminate form. In this case, the network is trivially both symmetric and antisymmetryc and the indeterminate form can be solved both as 0 or 1. Even if the relations considered in complex network theory are antireflexive, the definition of reciprocity explicitly excludes them frome the sums: this makes $r$ a valid index to measure the reciprocity of a generic relation (even in presence of its diagonal elements) because the diagonal of a relation has not influence on its symmetry \cite{DCM}. 

Now, it becomes clear why binary, undirected networks are particular cases of binary, directed networks: they represent symmetric relations where all links have their own reciprocal: $a_{ij}=a_{ji},\:\forall\:i\neq j$ or $A=A^T$.

It is possibile to count how many relations (networks) are of the three kinds, above: there are $2^{N}\cdot2^{\frac{N(N-1)}{2}}$ symmetric relations, $2^{N}\cdot3^{\frac{N(N-1)}{2}}$ antisymmetric relations (note that the $N$ diagonal elements are counted in both) and the remaining are asymmetric.
\newline
\newline
\indent So, the grandcanonical ensemble of binary, undirected networks (being a subset of the set of relations definable on $N$ vertices) is obtainable by simply requiring that the networks (relations) are undirected and without self-loops (that is, symmetric and antireflexive), restricting the number of graphs to be $2^{\frac{N(N-1)}{2}}$. On the other side, the grandcanonical ensemble of binary, directed networks is obtainable by simply requiring that the network (relations) are directed and without self-loops (that is, antireflexive), restricting the number of graphs to be $2^{N(N-1)}$.

\subsection{Transitivity}

\newtheorem{rel4}{Definition}
\begin{rel}
A relation is said to be transitive when $(v_{i}\mathcal{R}v_{j})\wedge(v_{j}\mathcal{R}v_{k})\:\Longrightarrow (v_{i}\mathcal{R}v_{k}),\:\forall\:i,j,k$.
\end{rel}

It is evident that the following, particular cases are included in the definition of transitivity \cite{DCM}:

$$
(v_{i}\mathcal{R}v_{j})\wedge(v_{j}\mathcal{R}v_{i})\:\Longrightarrow (v_{i}\mathcal{R}v_{i}),\:i\neq j,
$$
$$
(v_{j}\mathcal{R}v_{i})\wedge(v_{i}\mathcal{R}v_{j})\:\Longrightarrow (v_{j}\mathcal{R}v_{j}),\:i\neq j,
$$

\noindent implying that pairs of nodes mutually interacting should also interact with themselves, by means of a pair of self-loops. However if, as usual in network theory, antireflexive relations were considered, no vertex would interact with itself thus making a symmetric relation automatically not-transitive. To avoid this, we can restrict the definition of trasitivity to the triples of distinct vertices: $i, j, k$ with $i\neq j,\:j\neq k,\:i\neq k$. In this case a quantity can be defined to measure the degree of transitivity of the considered relation \cite{giorgio}, as shown in fig. \ref{clustdir}:

$$
c_{i}^{mid}=\frac{\sum_{k(\neq j)}\sum_{j}a_{ji}a_{ik}a_{jk}}{k_i^{in}k_i^{out}-k_{i}^{\leftrightarrow}}
$$

\noindent (where the indexes run from 1 to $N$). The symbols above refer to standard quantities in network theory: the \emph{out-degree} of node $i$, $k_{i}^{out}=\sum_{j(\neq i)} a_{ij}$, that is the number of outgoing links from $i$; the \emph{in-degree} of node $i$, $k_{i}^{in}=\sum_{j(\neq i)} a_{ji}$, that is the number of incoming links to $i$; the \emph{reciprocal degree} of node $i$, $k_{i}^{\leftrightarrow}=\sum_{j(\neq i)}a_{ij}a_{ji}$, that is the number of links outgoing from $i$ and having a reciprocal partner. In the case of a symmetryc relation, $A=A^T$ and $k_{i}^{in}=k_{i}^{out}=k_{i}^{\leftrightarrow}\equiv k_{i}$, the last symbol defining the \emph{degree of a node}, that is the number of its neighbors; so, in this case,

$$
c_{i}^{mid}=\frac{\sum_{k(\neq j)}\sum_{j}a_{ji}a_{ik}a_{jk}}{k_i^{in}k_i^{out}-k_{i}^{\leftrightarrow}}=\frac{\sum_{k(\neq j)}\sum_ja_{ij}a_{jk}a_{ki}}{k_i (k_i-1)}\equiv c_{i}
$$

\noindent and we recover the usual definition of the \emph{undirected clustering coefficient} \cite{holl}.

\begin{figure}[h!]
\includegraphics[scale=0.2]{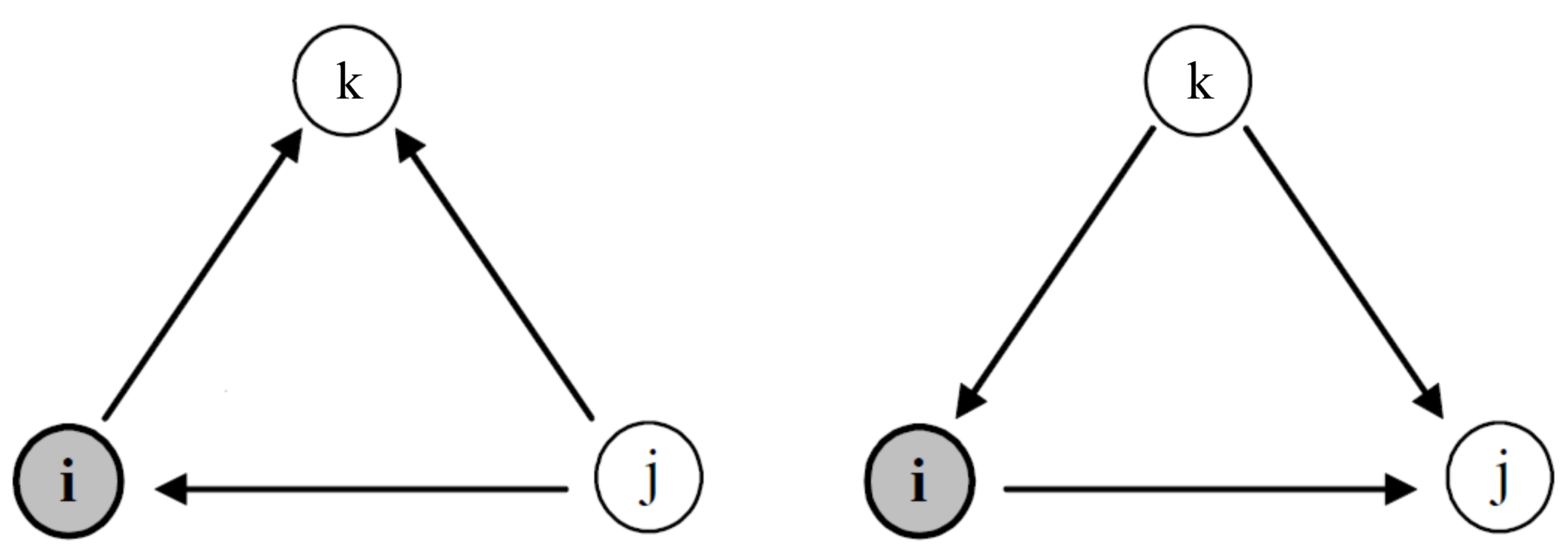}
\caption{Two examples of the paths measured by $c_{mid}.$ The grey vertex, $i$, is assumed as the point of view from which the measures are taken.}
\label{clustdir}
\end{figure}

To understand how it works, let us consider a symmetric relation (that is, a binary, undirected network). The degree of a node is the number of its neighbors: so, $\frac{k_i(k_i-2)}{2}$ is the maximum number of $i$'s neighbors' pairs that could be connected or, equivalently, the maximum number of triangles having $i$ as a vertex. Now, an undirected triangle is a graphical representation of a transitive relation: since the clustering coefficient counts the fraction of triangles effectively completed, it is also a measure of the fraction of transitive relations effectively realized. 

For binary, directed networks the only difference lies in what follows: because of the directionality of the links, the order of the considered vertices is not irrelevant. $c_{mid}$ counts the number of transitive relations by taking the point of view of the central (``middle'') vertex of the triangle (see fig. \ref{trans} and table \ref{tabtrans} for a visual example). As with the trace and the reciprocity we have three cases:

\begin{displaymath}
\left\{ \begin{array}{ll}
c_{i}^{mid}=0,\:\forall\:i: & \mbox{\emph{antitransitive}}.\\
\exists\:i:\:0<c_{i}^{mid}<1: & \mbox{\emph{intransitive}}.\\
c_{i}^{mid}=1,\:\forall\:i: & \mbox{\emph{transitive}}.
\end{array} \right.
\end{displaymath}

If the denominator is zero for some vertex, there are no relations to realize. In the binary, undirected case, this amounts to consider isolated nodes. In the directed case, the evenience of having a vanishing denominator is more subtle (see fig. \ref{trans}). In both these cases, we would have the indeterminate form $c_i^{mid}=0/0$, for some $i$. As for the reciprocity, these can be solved both as 0 and 1.

\begin{figure}[h!]
\includegraphics[scale=0.29]{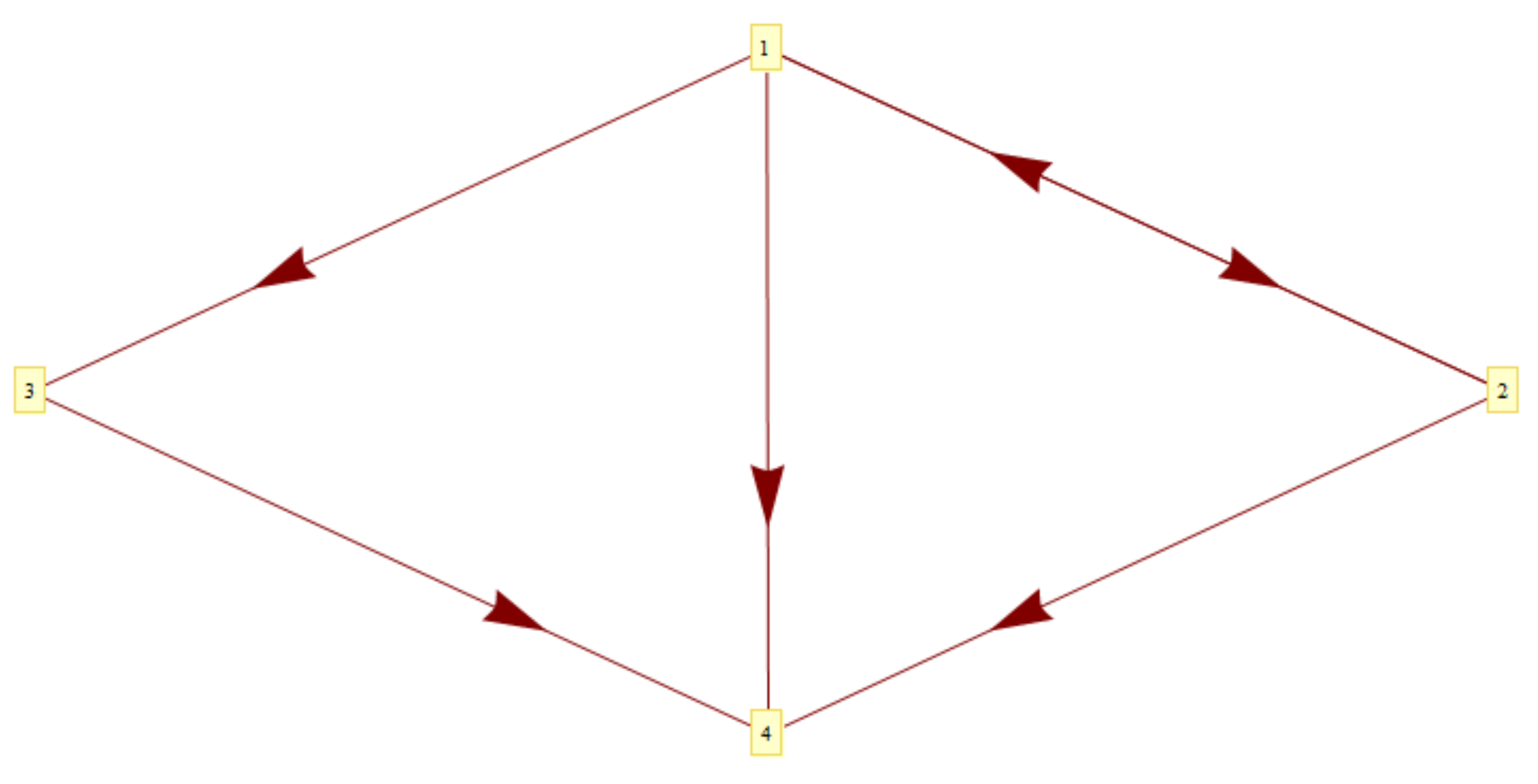}
\caption{A binary, directed network with $N=4$ shows the effectiveness of the index $c_{mid}$ in measuring the transitivity of the considered relation. Vertex 1 completes only one of the two potential transitive relations of which it is a vertex: the triangle 214. The triangle 213 is not completed. In fact, $v_2\mathcal{R}v_1 \wedge v_1\mathcal{R}v_4\Longrightarrow v_2\mathcal{R}v_4$; on the other side, $v_2\mathcal{R}v_1 \wedge v_1\mathcal{R}v_3$ but $v_2\overline{\mathcal{R}}v_3$. Vertex $v_4$ does not partecipate to any potential transitive relation. The considered relation is intransitive, because of vertex 1.}
\label{trans}
\end{figure}

\begin{table}[h!]
\centering
\begin{tabular}{|c|c|c|c|c|}
\hline
& $k_{i}^{out}$ & $k_{i}^{in}$ & $k_{i}^{\leftrightarrow}$ & $c_{i}^{mid}$\\
\hline
\hline
$v_1$ & 3 & 1 & 1 & 1/2\\
\hline
$v_2$ & 2 & 1 & 1 & 1\\
\hline
$v_3$ & 1 & 1 & 0 & 1\\
\hline
$v_4$ & 0 & 3 & 0 & $0/0$\\
\hline
\end{tabular}
\caption{Degree sequences and $c_{mid}$ of the four vertices forming the binary, directed network in fig. \ref{trans}.}
\label{tabtrans}
\end{table}

If the considered relation is symmetric, then the antitransitive request (a binary, undirected network without cycles) defines a \emph{tree}.

\section{The grandcanonical ensemble as a boolean algebra}

The adjacency matrix elements, $a_{ij}$s, can assume only the values $0$ and $1$. In the previous sections they were considered as indicators of some kind of relation. In this section, on the other side, they will be considered as \emph{boolean}, \emph{logic variables} \cite{TAGT}. So a binary, undirected network is a collection of $\frac{N(N-1)}{2}$ boolean variables (and $N(N-1)$ for a binary, directed network).

The grandcanonical ensemble of binary networks was obtained by the tree generation process: let us call it $\mathcal{G}^{B}_{N}$ to indicate the presence of $N$ nodes and only binary ($B$) relations. Until now, we have characterized the single graphs in $\mathcal{G}^{B}_{N}$ as relations. Now, let us characterize the ensemble $\mathcal{G}_N^B$ as a whole. 

\newtheorem{rel5}{Definition}
\begin{rel}
A boolean algebra is a sextuple $<\mathcal{B}, \cap, \cup, \neg, 0, 1>$, consisting in a set $\mathcal{B}$, equipped with two binary, internal operation, $\cap$ and $\cup$, and an unary operation, $\neg$, and two elements, $0$ and $1$, belonging to $\mathcal{B}$, such that the following axioms hold for all the elements, $b$, in $\mathcal{B}$: 1) $\cap$ and $\cup$ are associative, 2) $\cap$ and $\cup$ are commutative, 3) $\cap$ and $\cup$ are mutually distributive, 4) $\cap$ and $\cup$ have their own neutral elements, $1$ and $0$, such that $b\cap 1=b$ and $b\cup 0=b$, 5) each element in $\mathcal{B}$ has its own complement, $\neg b$, such that $b\cap\neg b=0$ and $b\cup\neg b=1$.
\end{rel}

Let us consider the following \emph{binary} operation \cite{ieee,CGT,GT,IGT,bb}

$$
\cap:\mathcal{G}^{B}_{N}\times\mathcal{G}^{B}_{N}
$$

\noindent known as \emph{intersection of graphs}, such that

$$
G_{1}\cap G_{2}\equiv(V_{1}\cap V_{2},\:E_{1}\cap E_{2}),\:\forall\:G_1,\:G_2\in\mathcal{G}^{B}_{N}.
$$

Now, since we have $V_{1}=V_{2}\equiv V$ the intersection becomes $G_{1}\cap G_{2}=(V,\:E_{1}\cap E_{2})$, acting as the intersection between the two sets of links \cite{CGT,GT,IGT,bb}. This operation has to be understood as acting between the pairs of nodes indexed by the same numbers, $i$ and $j$, when different graphs are considered. So, the intersection of graphs is, actually, the result of the $\frac{N(N-1)}{2}$ (or $N(N-1)$) intersections between corresponding pairs. So, given the pair of nodes $ij$ and two different graphs, $G_1$ and $G_2$, the elements $a_{ij}^1\in E_1$ and $a_{ij}^2\in E_2$ intersect as usual, binary variables as shown in table \ref{tab1} (see fig. \ref{intersect} for an example).

\begin{table}[h!]
\centering
\begin{tabular}{|c|c|c|}
\hline
$a_{ij}^{1}\in E_1$ & $a_{ij}^{2}\in E_2$ & $a_{ij}^{1}\cap a_{ij}^{2}\in E_1\cap E_2$ \\
\hline
\hline
1 & 1 & 1 \\
\hline
1 & 0 & 0 \\
\hline
0 & 1 & 0 \\
\hline
0 & 0 & 0 \\
\hline
\end{tabular}
\caption{Table of truth for the intersection of the boolean elements of the adjacency matrix $A$.}
\label{tab1}
\end{table}

\begin{figure}[h!]
\includegraphics[scale=0.3]{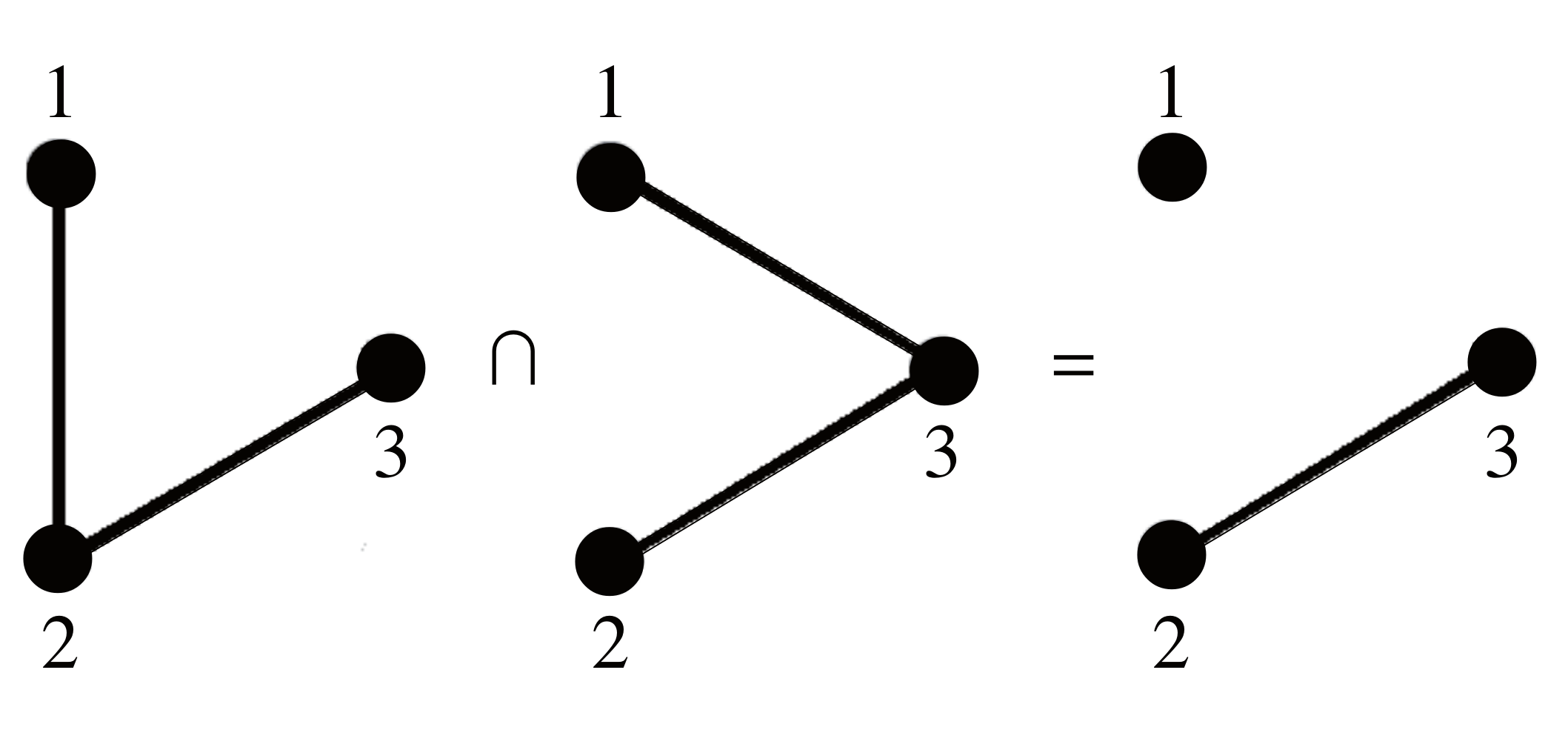}
\caption{Example of intersection of graphs, according to the table of truth \ref{tab1}.}
\label{intersect}
\end{figure}

Now, let us consider a second \emph{binary} operation \cite{CGT,GT,IGT,bb}:

$$
\cup:\mathcal{G}^{B}_{N}\times\mathcal{G}^{B}_{N}
$$

\noindent known as \emph{union of graphs}:

$$
G_{1}\cup G_{2}\equiv(V_{1}\cup V_{2},\:E_{1}\cup E_{2}),\:\forall\:G_1,\:G_2\in\mathcal{G}^{B}_{N}
$$

\noindent and since we have $V_{1}=V_{2}\equiv V$, the intersection becomes $G_{1}\cup G_{2}=(V,\:E_{1}\cup E_{2})$ acting as the union of the two sets of links \cite{CGT,GT,IGT,bb}. Also this operation is defined as acting between corresponding pairs of nodes: given the pair $ij$ and two different graphs, $G_1$ and $G_2$, the union between the elements $a_{ij}^1\in E_1$ and $a_{ij}^2\in E_2$ works as shown in table \ref{tab2} (see fig. \ref{union} for an example).

\begin{table}[h]
\centering
\begin{tabular}{|c|c|c|}
\hline
$a_{ij}^{1}\in E_1$ & $a_{ij}^{2}\in E_2$ & $a_{ij}^{1}\cup a_{ij}^{2}\in E_1\cup E_2$ \\
\hline
\hline
1 & 1 & 1 \\
\hline
1 & 0 & 1 \\
\hline
0 & 1 & 1 \\
\hline
0 & 0 & 0 \\
\hline
\end{tabular}
\caption{Table of truth for the union of the boolean elements of the adjacency matrix $A$.}
\label{tab2}
\end{table}

\begin{figure}[h!]
\includegraphics[scale=0.3]{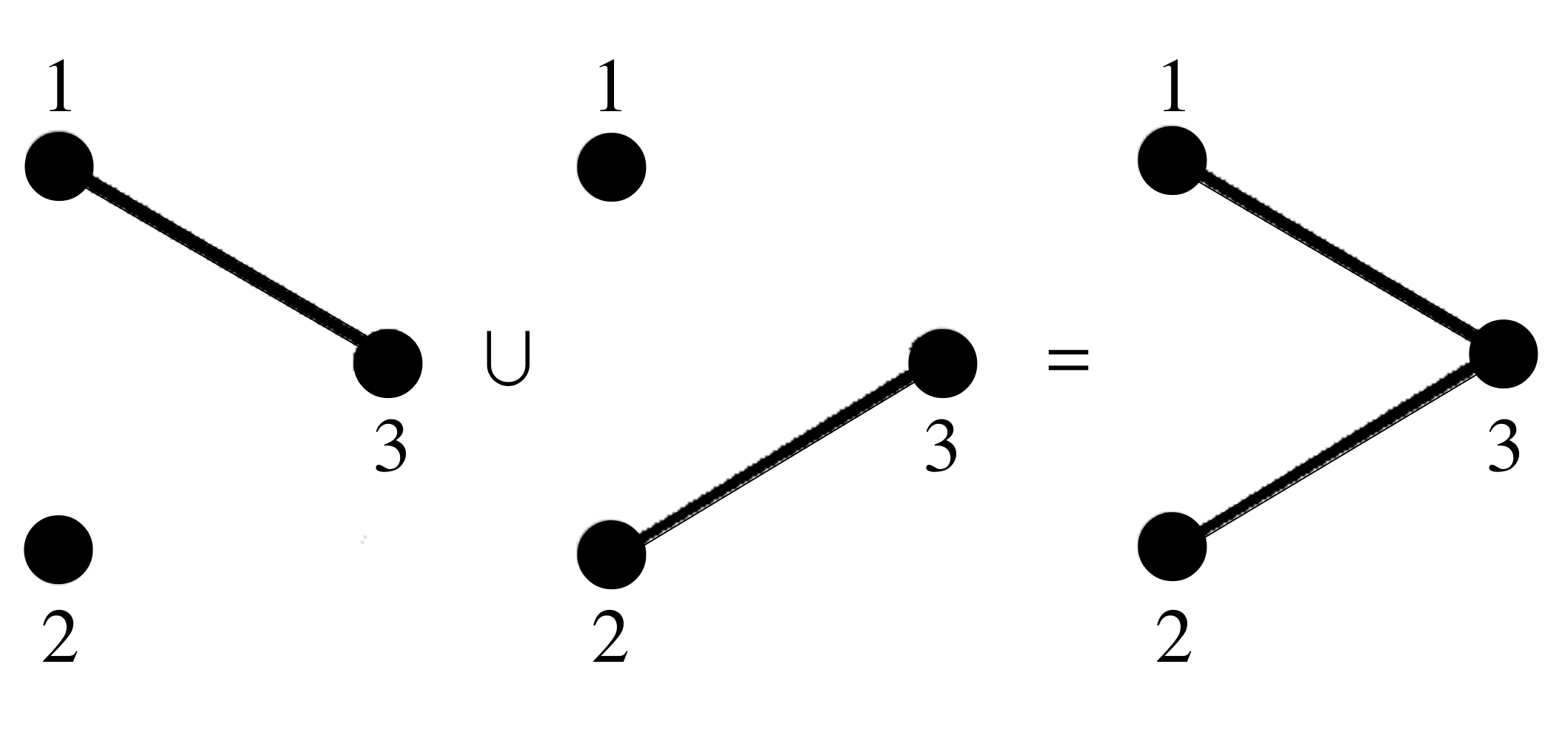}
\caption{Example of union of graphs, according to the table of truth \ref{tab2}.}
\label{union}
\end{figure}

The last operation we consider is an \emph{unary} operation acting on single graphs, the \emph{complement of a graph} \cite{CGT,GT,IGT,bb}:

$$
\neg:\mathcal{G}_{N}^{B},
$$

\noindent such that $\neg G\equiv(V,\:\neg E),\:\forall\:G\in\mathcal{G}^{B}_{N}$. The complement of the link set can be defined as the result of the complement operation on the single adjacency matrix elements, as shown by the rules in table \ref{tab3} (see fig. \ref{neg} for an example).

\begin{table}[h!]
\centering
\begin{tabular}{|c|c|}
\hline
$a_{ij}\in G$ & $\neg a_{ij}\in \neg G$ \\
\hline
\hline
1 & 0 \\
\hline
0 & 1 \\
\hline
\end{tabular}
\caption{Table of truth for the complement of the boolean elements of the adjacency matrix $A$.}
\label{tab3}
\end{table}

\begin{figure}[h!]
\includegraphics[scale=0.3]{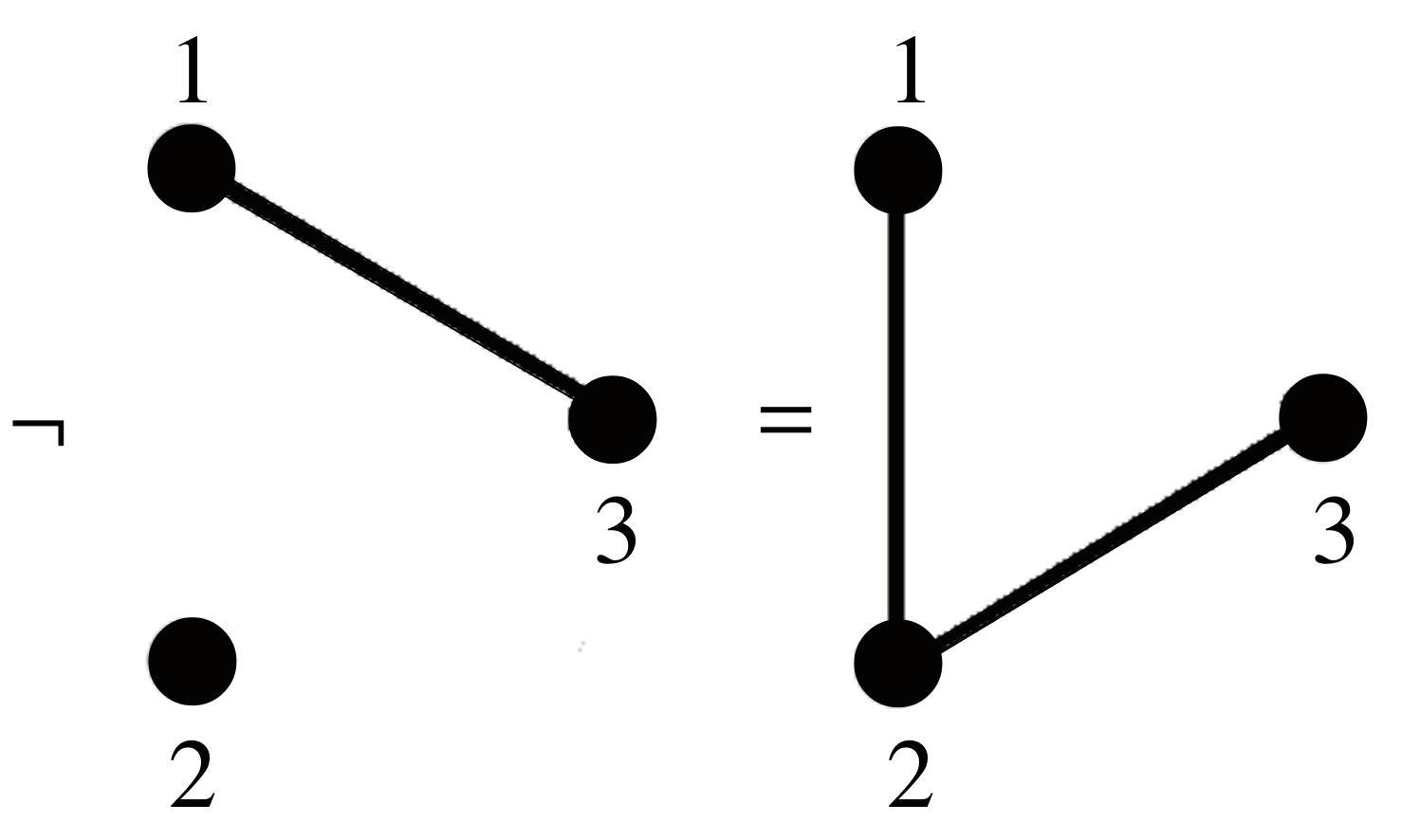}
\caption{Example of the complement of a graph, according to the table of truth \ref{tab3}.}
\label{neg}
\end{figure}

\newtheorem{thm}{Theorem}
\begin{thm}
The grandcanonical ensemble $\mathcal{G}_N^B$ of binary networks (undirected or directed), equipped with the operations od intersection, union and complement of a graph, is a boolean algebra.
\end{thm}

\begin{proof}
It is easy to verify that the intersection of graphs is an \emph{internal} operation, that is $\cap:\mathcal{G}^{B}_{N}\times\mathcal{G}^{B}_{N}\longrightarrow\mathcal{G}^{B}_{N}$. In fact, being the results of the intersection between corresponding pairs of nodes only $0$ or $1$, it represents a binary element whose contribution to the resulting graph is automatically accounted for in the tree generation. The same is valid for the union of graphs: $\cup:\mathcal{G}^{B}_{N}\times\mathcal{G}^{B}_{N}\longrightarrow\mathcal{G}^{B}_{N}$. The complement of a graph is an \emph{internal} operation, too: $\neg:\mathcal{G}_{N}^{B}\longrightarrow \mathcal{G}_{N}^{B}$. This is even more evident than for the precedent operations: it is sufficient to consider that, at every step of the tree generation, a choice between $0$ and its complement, $1$, has to be made. So, the complement of every choice is accounted for by this process.

The neutral element for the intersection is the complete graph with $N$ vertices, $K_{N}$: in facts, $G\cap K_{N}=K_{N}\cap G=G,\:\forall\:G\in\mathcal{G}^{B}_{N}$. In fact, all the adjacency matrix elements of $K_N$ are 1, so $1\cap 0=0$ and $1\cap 1=1$. 

The neutral element of the union is the empty graph with $N$ vertices, $E_{N}$: $G\cup E_{N}=E_{N}\cup G=G,\:\forall\:G\in\mathcal{G}_{N}^{B}$. In fact, all the adjacency matrix elements of $E_N$ are 0, so $0\cup 0=0$ and $0\cup 1=1$. 

Note that

$$
G\cap\neg G=E_{N},\:G\cup\neg G=K_{N},\forall\:G\in\mathcal{G}_{N}^{B}
$$

\noindent because these operations reduce to operate between pairs of nodes and it is evident that $1\cup\neg 1=1$, $0\cup\neg 0=1$, $1\cap\neg 1=0$ and $0\cap\neg 0=0$.

These operations have also the other requested properties. When computed on a single pair of variables, the intersection and the union are commutative, associative, mutually distributive \cite{CGT,GT,IGT,bb}. When dealing with graphs, all the pairs have to be considered independently from each other. So these propeties are preserved also for graphs: the intersection and the union are commutative

$$
G_1\cap G_2=G_2\cap G_1,\:G_1\cup G_2=G_2\cup G_1,\:\forall\:G_1,\:G_2\in\mathcal{G}^{B}_{N},
$$

\noindent associative

$$
G_1\cap (G_2\cap G_3)=(G_1\cap G_2)\cap G_3,\:\forall\:G_1,\:G_2\in\mathcal{G}^{B}_{N},
$$
$$
G_1\cup (G_2\cup G_3)=(G_1\cup G_2)\cup G_3,\:\forall\:G_1,\:G_2\in\mathcal{G}^{B}_{N}
$$

\noindent and mutually distributive

$$
G_1\cap (G_2\cup G_3)=(G_1\cap G_2)\cup(G_1\cap G_3),\:\forall\:G_1, G_2, G_3\in\mathcal{G}^{B}_{N},
$$
$$
G_1\cup (G_2\cap G_3)=(G_1\cup G_2)\cap(G_1\cup G_3),\:\forall\:G_1, G_2, G_3\in\mathcal{G}^{B}_{N}.
$$
\end{proof}

So, by enriching $\mathcal{G}_{N}^{B}$ with these three operations we have a \emph{boolean algebra of binary, undirected (or directed) graphs}, having $N$ vertices each:

\begin{displaymath}
\mathsf{G^{B}_N}\equiv<\mathcal{G}_{N}^B, \cap, \cup, \neg, K_{N}, E_{N}>.
\end{displaymath}

We have made no distinction between undirected and directed graphs during the proofs because they can be treated in exactly the same way. So we can end up with a boolean algebra of binary undirected or directed graphs, depending on the type of network chosen.

\section{The grandcanonical ensemble as a sigma-algebra}

From the operations defined above, we can consider every generic graph $G$ in $\mathcal{G}_{N}^B$ as a ``subset'' of $K_N$, in the sense that $E \subseteq E_{K_N},\:\forall\:G\in\mathcal{G}_{N}^B$, where $G=(V,\:E)$, that is, the link set of every graph is a subset of the link set of the complete graph. In fact, by using the intersection operation

$$
G\cap K_{N}=G,\:\forall\:G\in\mathcal{G}_{N}^B.
$$

\newtheorem{thm2}{Theorem}
\begin{thm}
The grandcanonical ensemble $\mathcal{G}_{N}^B$ is the sigma-algebra generated by the subsets of the complete graph $K_N$ (its power-set): the so-called discrete sigma-algebra, whose cardinality is $2^{\frac{N(N-1)}{2}}$ for binary, undirected networks and $2^{N(N-1)}$ for binary, directed networks.
\end{thm}

\begin{proof}
In the first place, $K_{N}\in\mathcal{G}_{N}^B$: this is evident from the tree generation process.

In the second place, $G\in\mathcal{G}_{N}^B\Longrightarrow\neg G\in\mathcal{G}_{N}^B$: in fact, every graph has its complement inside the grandcanonical ensemble.

In the third place, $G_i\in\mathcal{G}_{N}^B,\:\forall\:i=1\dots n\Longrightarrow\bigcup_{i=1}^n G_i\in\mathcal{G}_{N}^B$, that is: the union of a finite familiy of graphs, belonging to the grandcanonical ensemble, still belongs to it. This can be verified by taking, first, the union of two graphs (which is an internal operation, as seen before) and then considering a third one and so on. Of course, no infinite union can be considered, because the cardinality of the grandcanonical ensemble is finite.
\end{proof}

With a new symbol, we can indicate this sigma-algebra as

$$
\mathbb{G}_{N}^B\equiv(K_N,\:\mathcal{G}_{N}^B).
$$

\section{Conclusions}

In statistical mechanics of networks, the set of all graphs with the same number of vertices and number of links varying from zero to the maximum (the support of the grandcanonical distribution, known, in physics, as the \emph{grandanocanical ensemble}) is used as a tool to carry on the calculations of the partition functions. However, little attention has been paid to its mathematical structure. In this paper, an answer from a logic-algebraic point of view has been provided. 

Even if the extension of the logical operations as the union, the intersection and the complement between boolean variables to binary graphs was already known, very little attention has been paid to the role of these operations in characterizing algebraic structures of binary graphs.

By enriching the grandcanonical ensemble with the three operations of \emph{intersection}, \emph{union} and \emph{complement of graphs}, we have found the structure of a \emph{boolean algebra}, that is a set of elements closed under the aforementioned operations and including the neutral elements of the intersection and the union. 

Moreover, by considering the graphs belonging to the grandcanonical ensemble as subsets of the complete graph $K_N$ (which is true for the corresponding link sets) and extending the operation of union also to finite sets of graphs, the grandcanonical ensemble becomes also a \emph{sigma-algebra}: in particular, the so-called \emph{discrete sigma-algebra}, generated by the power-set of $K_N$.

A natural extension of this analysis concerns the weighted networks, on the one side, and the characterization of the microcanonical ensemble, on the other. One of the major challenges about the weighted networks is to find the right generalization of the binary, logic operations to the weighted grandcanonical ensemble; the microcanonical ensemble, instead, forces us to face the need of defining a novel algebraic structure to take into account the hard, topological constraints that are usually imposed over it (as the degree sequence, for example).

\begin{acknowledgments}
T. S. acknowledges support from an ERC Advanced Investigator Grant.
\end{acknowledgments}

\end{document}